\documentclass[submission,copyright,creativecommons]{eptcs}
\providecommand{\event}{QPL 2019} 
\usepackage{underscore}           

\title{Deriving Dagger Compactness}
\def\titlerunning{Deriving Dagger Compactness}
\author{Sean Tull\thanks{This research was supported by an EPSRC Doctoral Prize.}
\institute{University of Oxford}
\email{sean.tull@cs.ox.ac.uk}
}
\def\authorrunning{S. Tull}
\usepackage{amsmath,amssymb,stmaryrd,xspace,enumerate}
\usepackage{amsthm}
\usepackage{tikz-cd}
\usepackage{multirow}
\usepackage{mathtools}
\usepackage{url}
\usepackage{relsize}
\usepackage{centernot}
\usepackage{bm} 
\usepackage[inline]{enumitem}
\setenumerate{label=\arabic*.}

\usepackage{tikzit}

\newcommand{\dsubcat}{\catD}
\newcommand{\paxiom}[1]{\textbf{(#1)}}

\theoremstyle{plain}

\newtheorem{theorem}{Theorem}
\newtheorem{proposition}[theorem]{Proposition}

\newtheorem{lemma}[theorem]{Lemma}

\theoremstyle{definition}
\newtheorem{definition}[theorem]{Definition}

\newtheorem{examples}[theorem]{Examples}
\newtheorem{remark}[theorem]{Remark}

\newcommand{\peassign}[1]{\overline{#1}} 

\newcommand{\notetoself}[1]{}
\newcommand{\omitfornow}[1]{}
\newcommand{\omitthis}[1]{}

\newcommand{\Rel}{\cat{Rel}}

\newcommand{\Class}{\cat{FClass}}

\newcommand{\FCStar}{\cat{FCStar}}

\newcommand{\Quant}[1]{\cat{Quant}_{#1}}

\newcommand{\Mat}{\cat{Mat}} 
\newcommand{\MatS}{\Mat_S}

\newcommand{\CPM}{\ensuremath{\mathsf{CPM}}\xspace}
\newcommand{\Split}[1]{\ensuremath{\mathrm{Split}}\xspace}

\newcommand{\cat}[1]{\ensuremath{\mathbf{#1}}\xspace}

\newcommand{\catD}{\cat{D}}
\newcommand{\catA}{\cat{A}}
\newcommand{\catC}{\cat{C}}
\newcommand{\catB}{\cat{B}}

\newcommand{\pure}{\mathsf{pure}}

\newcommand{\id}[1]{\ensuremath{\mathrm{id}_{#1}}}

\newcommand{\discard}[1]{\ensuremath{\tinygroundnew_{#1}}}

\newcommand{\discardflip}[1]{\ensuremath{\tinygroundflipnew_{#1}}}

\DeclareFontFamily{U}{mathux}{\hyphenchar\font45}
\DeclareFontShape{U}{mathux}{m}{n}{
      <5> <6> <7> <8> <9> <10>
      <10.95> <12> <14.4> <17.28> <20.74> <24.88>
      mathux10
      }{}
\DeclareSymbolFont{mathux}{U}{mathux}{m}{n}

\DeclareMathSymbol{\bigovee}{1}{mathux}{"8F}

\DeclareMathSymbol{\bigperp}{1}{mathux}{"4E}

\DeclareFontFamily{U}{mathb}{\hyphenchar\font45}
\DeclareFontShape{U}{mathb}{m}{n}{
      <5> <6> <7> <8> <9> <10> gen * mathb
      <10.95> mathb10 <12> <14.4> <17.28> <20.74> <24.88> mathb12
      }{}
\DeclareSymbolFont{mathb}{U}{mathb}{m}{n}
\DeclareFontSubstitution{U}{mathb}{m}{n}
\DeclareMathSymbol{\mylgroup}{\mathbin}{mathb}{'160}
\DeclareMathSymbol{\myrgroup}{\mathbin}{mathb}{'161}

\DeclareMathOperator{\Tr}{Tr}

\newcommand{\isomto}{\xrightarrow{\sim}}

\newcommand{\img}{\mathrm{im}}

\newcommand{\coker}{\mathrm{coker}}

\newcommand{\plusI}[1]{\mathsf{GP}({#1})}

\newcommand{\hilbH}{\mathcal{H}} 

\usepackage{tikz,xypic}

\usetikzlibrary{decorations.pathreplacing,decorations.markings,arrows.meta,backgrounds}
\pgfdeclarelayer{edgelayer}
\pgfdeclarelayer{nodelayer}
\pgfsetlayers{background,edgelayer,nodelayer,main}

\tikzstyle{cdot}=[circle, draw=black, fill=black!25, inner sep=.4ex] 
\tikzstyle{bigdot}=[dot, inner sep=0pt]
\tikzstyle{whitedot}=[circle, draw=black, fill=white, inner sep=.4ex]
\tikzstyle{greydot}=[circle, draw=black, fill=black!25, inner sep=.4ex] 
\tikzstyle{blackdot}=[circle, draw=black, fill=black, inner sep=.4ex]
\tikzset{arrow/.style={decoration={
    markings,
    mark=at position #1 with \arrow{>[length=2pt, width=3pt]}},
    postaction=decorate},
    reverse arrow/.style={decoration={
    markings,
    mark=at position #1 with {{\arrow{<[length=2pt, width=3pt]}}}},
    postaction=decorate}
}

\newcommand{\tinycup}{\smash{\raisebox{2pt}{\hspace{-2pt}\ensuremath{\begin{pic}[scale=0.2]
   \pgftransformscale{1.5} \draw[arrow=.6, scale = 1] (0,0) to[out=-90,in=-90,looseness=1.5] (1.5,0);
\end{pic}}}}}

\newcommand{\tinycap}{\smash{\raisebox{-3pt}{\hspace{-2pt}\ensuremath{\begin{pic}[scale=0.2, yscale=-1]
   \pgftransformscale{1.5} \draw[arrow=.6, scale = 1] (0,0) to[out=-90,in=-90,looseness=1.5] (1.5,0);
\end{pic}}}}}

\newcommand{\tinyswap}{
\smash{\raisebox{0pt}{\hspace{-2pt}\ensuremath{
   \begin{pic}[scale=.25]
  \draw (0,-.5) to[out=80,in=-100] (1,.5);
  \draw (1,-.5) to[out=100,in=-80] (0,.5);
\end{pic}}}}
}
\newenvironment{pic}[1][] {\begin{aligned}\begin{tikzpicture}[scale=2.0, font=\tiny,#1]}{\end{tikzpicture}\end{aligned}} 

\newif\ifvflip\pgfkeys{/tikz/vflip/.is if=vflip}
\newif\ifhflip\pgfkeys{/tikz/hflip/.is if=hflip}
\newif\ifhvflip\pgfkeys{/tikz/hvflip/.is if=hvflip}

\newenvironment{picc}[1][]
{\begin{aligned}\begin{tikzpicture}[font=\tiny,#1]}
{\end{tikzpicture}\end{aligned}}

\newlength\minimummorphismwidth
\newlength\stateheight
\newlength\minimumstatewidth
\newlength\connectheight
\tikzset{colour/.initial=white}

\tikzstyle{pure}=[line width=.7pt]

\makeatletter

\pgfdeclareshape{groundd}
{
    \savedanchor\centerpoint
    {
        \pgf@x=0pt
        \pgf@y=0pt
    }
    \anchor{center}{\centerpoint}
    \anchorborder{\centerpoint}

    \anchor{north}
    {
        \pgf@x=0pt
        \pgf@y=0.16\stateheight
    }
    \anchor{south}
    {
        \pgf@x=0pt
        \pgf@y=0pt
    }
    \saveddimen\overallwidth
    {
        \pgfkeysgetvalue{/pgf/minimum width}{\minwidth}
        \pgf@x=\minimumstatewidth
        \ifdim\pgf@x<\minwidth
            \pgf@x=\minwidth
        \fi
    }
    \backgroundpath
    {
        \begin{pgfonlayer}{main} 
        \pgfsetstrokecolor{black}
        \pgfsetlinewidth{1.25pt}
        \ifhflip
            \pgftransformyscale{-1}
        \fi
        \pgftransformscale{0.5}
        \pgfpathmoveto{\pgfpoint{-0.5*\overallwidth}{0pt}}
        \pgfpathlineto{\pgfpoint{0.5*\overallwidth}{0pt}}
        \pgfpathmoveto{\pgfpoint{-0.33*\overallwidth}{0.33*\stateheight}}
        \pgfpathlineto{\pgfpoint{0.33*\overallwidth}{0.33*\stateheight}}
        \pgfpathmoveto{\pgfpoint{-0.16*\overallwidth}{0.66*\stateheight}}
        \pgfpathlineto{\pgfpoint{0.16*\overallwidth}{0.66*\stateheight}}
        \pgfpathmoveto{\pgfpoint{-0.02*\overallwidth}{\stateheight}}
        \pgfpathlineto{\pgfpoint{0.02*\overallwidth}{\stateheight}}
        \pgfusepath{stroke}
        \end{pgfonlayer}
    }
}

\usepackage{tikz,xypic}
\usetikzlibrary{decorations.pathreplacing,decorations.markings,arrows.meta,backgrounds,shapes}
\usetikzlibrary{circuits.ee.IEC}
\pgfdeclarelayer{edgelayer}
\pgfdeclarelayer{nodelayer}
\pgfsetlayers{background,edgelayer,nodelayer,main}
\tikzstyle{none}=[inner sep=0mm]
\tikzstyle{every loop}=[]
\tikzstyle{mark coordinate}=[inner sep=0pt,outer sep=0pt,minimum size=3pt,fill=black,circle]

\tikzset{arrow/.style={decoration={
    markings,
    mark=at position #1 with \arrow{>[length=2pt, width=3pt]}},
    postaction=decorate},
    reverse arrow/.style={decoration={
    markings,
    mark=at position #1 with {{\arrow{<[length=2pt, width=3pt]}}}},
    postaction=decorate}
}

\tikzstyle{upground}=[circuit ee IEC,thick,ground,rotate=90,scale=1.5]
\tikzstyle{downground}=[circuit ee IEC,thick,ground,rotate=-90,scale=1.5]

\newcommand{\mapminh}{5mm} 
\newcommand{\stateminh}{5mm}
\newcommand{\maplw}{0.7pt} 
\newcommand{\stateshift}{-0.2pt}
\newcommand{\effectshift}{-0.2pt}

\tikzstyle{box}=[map]
\tikzstyle{medium box}=[medium map]
\tikzstyle{dot}=[inner sep=0mm,minimum width=2mm,minimum height=2mm,draw,shape=circle]  
\tikzstyle{black dot}=[dot,fill=black]
\tikzstyle{white dot}=[dot,fill=white,,text depth=-0.2mm]
\tikzstyle{grey dot}=[dot,fill=black!25] 

\tikzstyle{corner1}=[box,fill=white, font=\footnotesize] %
\tikzstyle{corner2}=[dot,fill=white, font=\footnotesize] %
\tikzstyle{corner3}=[dot,fill=black!25, font=\footnotesize] %
\tikzstyle{corner4}=[dot,fill=black, font=\footnotesize] %

\tikzstyle{scalar}=[circle,draw,inner sep=2pt, line width=\maplw] 

\usetikzlibrary{shapes.misc, positioning}

\tikzset{stateshape/.style={append after command={
   \pgfextra
        \draw[sharp corners, fill=white, line width = \maplw]%
    (\tikzlastnode.west)%
    [rounded corners=0pt] |- (\tikzlastnode.north)%
    [rounded corners=0pt] -| (\tikzlastnode.east)%
    [rounded corners=5pt] |- (\tikzlastnode.south)%
    [rounded corners=5pt] -| (\tikzlastnode.west);
   \endpgfextra}}}

\tikzset{effectshape/.style={append after command={
   \pgfextra
        \draw[sharp corners, fill=white, line width = \maplw]%
    (\tikzlastnode.west)%
    [rounded corners=0pt] |- (\tikzlastnode.south)%
    [rounded corners=0pt] -| (\tikzlastnode.east)%
    [rounded corners=5pt] |- (\tikzlastnode.north)%
    [rounded corners=5pt] -| (\tikzlastnode.west);
   \endpgfextra}}}

 \tikzstyle{map}=[draw,shape=rectangle, inner sep=2pt,minimum height=\mapminh, minimum width=7mm,fill=white]

\tikzstyle{point}=[stateshape,inner sep=2pt, minimum width=6mm, minimum height=\stateminh, yshift=\stateshift]
\tikzstyle{copoint}=[effectshape,inner sep=.2pt, minimum width=6mm, minimum height=\stateminh, yshift=-\effectshift]

\tikzstyle{wide point}=[point, minimum width=12mm]
\tikzstyle{wide copoint}=[copoint, minimum width=12mm]

 \tikzstyle{map}=[draw,shape=rectangle, inner sep=2pt,minimum height=\mapminh, minimum width=7mm,fill=white, line width = \maplw]

\tikzstyle{medium map} = [map, minimum width = 12mm] 
\tikzstyle{semilarge map} = [map, minimum width = 15mm] 
\tikzstyle{large map} = [map, minimum width = 18mm] 

\tikzstyle{kpoint} =[point]
\tikzstyle{kpointadj} =[copoint]
\tikzstyle{kpointconj}=[dagpointconj] 

\makeatletter
\newcommand{\boxshape}[3]{%
\pgfdeclareshape{#1}{
\inheritsavedanchors[from=rectangle] 
\inheritanchorborder[from=rectangle]
\inheritanchor[from=rectangle]{center}
\inheritanchor[from=rectangle]{north}
\inheritanchor[from=rectangle]{south}
\inheritanchor[from=rectangle]{west}
\inheritanchor[from=rectangle]{east}
\backgroundpath{
\southwest \pgf@xa=\pgf@x \pgf@ya=\pgf@y
\northeast \pgf@xb=\pgf@x \pgf@yb=\pgf@y

\@tempdima=#2
\@tempdimb=#3

\pgfpathmoveto{\pgfpoint{\pgf@xa - 5pt + \@tempdima}{\pgf@ya}}
\pgfpathlineto{\pgfpoint{\pgf@xa - 5pt - \@tempdima}{\pgf@yb}}
\pgfpathlineto{\pgfpoint{\pgf@xb + 5pt + \@tempdimb}{\pgf@yb}}
\pgfpathlineto{\pgfpoint{\pgf@xb + 5pt - \@tempdimb}{\pgf@ya}}
\pgfpathlineto{\pgfpoint{\pgf@xa - 5pt + \@tempdima}{\pgf@ya}}
\pgfpathclose
}
}}

\boxshape{NEbox}{0pt}{3pt} 
\boxshape{SEbox}{0pt}{-3pt}
\boxshape{NWbox}{3pt}{0pt}
\boxshape{SWbox}{-3pt}{0pt}
\boxshape{rec-box}{0pt}{0pt}
\makeatother

\tikzstyle{cloud}=[shape=cloud,draw,minimum width=1.5cm,minimum height=1.5cm]

\tikzstyle{dagmap}=[draw,shape=NEbox,inner sep=2pt,minimum height=\mapminh,fill=white, line width = \maplw] %
\tikzstyle{dashedmap}=[draw,dashed,shape=NEbox,inner sep=2pt,minimum height=\mapminh,fill=white, line width = \maplw]
\tikzstyle{mapdag}=[draw,shape=SEbox,inner sep=2pt,minimum height=\mapminh,fill=white, line width = \maplw]
\tikzstyle{mapadj}=[draw,shape=SEbox,inner sep=2pt,minimum height=\mapminh,fill=white, line width = \maplw]
\tikzstyle{maptrans}=[draw,shape=SWbox,inner sep=2pt,minimum height=\mapminh,fill=white, line width = \maplw]
\tikzstyle{mapconj}=[draw,shape=NWbox,inner sep=2pt,minimum height=\mapminh,fill=white, line width = \maplw]

\tikzstyle{medium dagmap}=[draw,shape=NEbox,inner sep=2pt,minimum height=\mapminh,fill=white,minimum width=7mm, line width = \maplw]
\tikzstyle{semilarge dagmap}=[draw,shape=NEbox,inner sep=2pt,minimum height=\mapminh,fill=white,minimum width=9.5mm, line width = \maplw]
\tikzstyle{large dagmap}=[draw,shape=NEbox,inner sep=2pt,minimum height=\mapminh,fill=white,minimum width=12mm, line width = \maplw]

\makeatletter

\pgfdeclareshape{cornerpoint}{
\inheritsavedanchors[from=rectangle] 
\inheritanchorborder[from=rectangle]
\inheritanchor[from=rectangle]{center}
\inheritanchor[from=rectangle]{north}
\inheritanchor[from=rectangle]{south}
\inheritanchor[from=rectangle]{west}
\inheritanchor[from=rectangle]{east}
\backgroundpath{
\southwest \pgf@xa=\pgf@x \pgf@ya=\pgf@y
\northeast \pgf@xb=\pgf@x \pgf@yb=\pgf@y

\pgfmathsetmacro{\pgf@shorten@left}{\pgfkeysvalueof{/tikz/shorten left}}
\pgfmathsetmacro{\pgf@shorten@right}{\pgfkeysvalueof{/tikz/shorten right}}

\pgfpathmoveto{\pgfpoint{0.5 * (\pgf@xa + \pgf@xb)}{\pgf@ya - 5pt}}
\pgfpathlineto{\pgfpoint{\pgf@xa - 8pt + \pgf@shorten@left}{\pgf@yb - 1.5 * \pgf@shorten@left}}
\pgfpathlineto{\pgfpoint{\pgf@xa - 8pt + \pgf@shorten@left}{\pgf@yb}}
\pgfpathlineto{\pgfpoint{\pgf@xb + 8pt - \pgf@shorten@right}{\pgf@yb}}
\pgfpathlineto{\pgfpoint{\pgf@xb + 8pt - \pgf@shorten@right}{\pgf@yb - 1.5 * \pgf@shorten@right}}
\pgfpathclose
}
}

\pgfdeclareshape{cornercopoint}{
\inheritsavedanchors[from=rectangle] 
\inheritanchorborder[from=rectangle]
\inheritanchor[from=rectangle]{center}
\inheritanchor[from=rectangle]{north}
\inheritanchor[from=rectangle]{south}
\inheritanchor[from=rectangle]{west}
\inheritanchor[from=rectangle]{east}
\backgroundpath{
\southwest \pgf@xa=\pgf@x \pgf@ya=\pgf@y
\northeast \pgf@xb=\pgf@x \pgf@yb=\pgf@y

\pgfmathsetmacro{\pgf@shorten@left}{\pgfkeysvalueof{/tikz/shorten left}}
\pgfmathsetmacro{\pgf@shorten@right}{\pgfkeysvalueof{/tikz/shorten right}}

\pgfpathmoveto{\pgfpoint{0.5 * (\pgf@xa + \pgf@xb)}{\pgf@yb + 5pt}}
\pgfpathlineto{\pgfpoint{\pgf@xa - 8pt + \pgf@shorten@left}{\pgf@ya + 1.5 * \pgf@shorten@left}}
\pgfpathlineto{\pgfpoint{\pgf@xa - 8pt + \pgf@shorten@left}{\pgf@ya}}
\pgfpathlineto{\pgfpoint{\pgf@xb + 8pt - \pgf@shorten@right}{\pgf@ya}}
\pgfpathlineto{\pgfpoint{\pgf@xb + 8pt - \pgf@shorten@right}{\pgf@ya + 1.5 * \pgf@shorten@right}}
\pgfpathclose
}
}

\makeatother

\pgfkeyssetvalue{/tikz/shorten left}{0pt}
\pgfkeyssetvalue{/tikz/shorten right}{0pt}

\tikzstyle{dagpoint common}=[draw,fill=white,inner sep=1pt, line width = \maplw, minimum height = 4mm, yshift=1.2pt] 
\tikzstyle{dagpoint sc}=[shape=cornerpoint,dagpoint common]
\tikzstyle{dagpoint adjoint sc}=[shape=cornercopoint,dagpoint common]
\tikzstyle{dagpoint}=[shape=cornerpoint,shorten left=4pt,dagpoint common]
\tikzstyle{dagpointadj}=[shape=cornercopoint,shorten left=5pt,dagpoint common]
\tikzstyle{dagpointconj}=[shape=cornerpoint,shorten right=5pt,dagpoint common]
\tikzstyle{dagpointtrans}=[shape=cornercopoint,shorten right=5pt,dagpoint common]
\tikzstyle{dagpointsymm}=[shape=cornerpoint,shorten left=5pt,shorten right=5pt,dagpoint common]

\tikzstyle{widedagpoint}=[dagpoint, minimum width=1 cm, inner sep=2pt]
\tikzstyle{widedagpointadj}=[dagpointadj, minimum width=1 cm, inner sep=2pt]

\tikzstyle{every picture}=[baseline=-0.25em,scale=0.5]
\tikzstyle{label}=[font=\footnotesize,text height=1ex, text depth=0.15ex]
\tikzstyle{math}=[text height=1ex, text depth=0.15ex]

\usetikzlibrary[shapes]

\tikzset{
sidetriangle/.style = {regular polygon, regular polygon sides = 3, aspect = 1, shape border rotate = 90, draw, inner sep = 0, minimum width = 1.2cm}
}

\tikzset{
isoc/.style = {shape=isosceles triangle, shape border rotate = 180, isosceles triangle stretches = true, minimum width = 1.2cm, minimum height= 1.5cm, inner sep = 0.3}}

\tikzset{
coarse/.style = {shape = circle, fill = white, draw, inner sep = 0, minimum width =0.125cm}
}
\tikzset{
coarsesymbol/.style = {shape = circle, fill = white, inner sep = -0.7, minimum width = 0.125cm}
}

\tikzstyle{sidetriangle2}=[sidetriangle, minimum width = 2cm, fill=white]
\tikzstyle{sideisocsmall}]=[style=isoc, minimum width = 1cm, minimum height = 0.8cm, draw, fill=white, font=\Large]
\tikzstyle{sideisoc}]=[style=isoc, minimum width = 2cm, draw, fill=white, font=\Large]
\tikzstyle{sideisocmid}]=[style=isoc, minimum width = 2.5cm, draw, fill=white, font=\Large]
\tikzstyle{sideisocmedium}]=[style=isoc, minimum width = 3cm, draw, fill=white, font=\Large]

\newcommand{\tinygroundnew}{
\smash{
{\hspace{-3pt}
\ensuremath{
\begin{picc}[scale=1.0] 
    \node[upground, xscale=0.8, yscale=0.7] (1) at (0,0.16) {};
    \draw (0,0.03) to (0,-0.25);
\end{picc}
}\hspace{-1pt}}}}

\newcommand{\tinygroundflipnew}{
\smash{
{\hspace{-3pt}
\ensuremath{
\begin{picc}[yscale=-1.0] 
    \node[upground, xscale=0.8, yscale=-0.7] (1) at (0,0.10) {};
    \draw (0,-0.03) to (0,-0.31);
\end{picc}
}\hspace{-1pt}}}}

\boxshape{NEbox}{0pt}{3pt} 
\boxshape{SEbox}{0pt}{-3pt}
\boxshape{NWbox}{3pt}{0pt}
\boxshape{SWbox}{-3pt}{0pt}
\boxshape{rec-box}{0pt}{0pt}
\makeatother

\tikzstyle{cloud}=[shape=cloud,draw,minimum width=1.5cm,minimum height=1.5cm]

\tikzstyle{dagmap}=[draw,shape=NEbox,inner sep=2pt,minimum height=\mapminh,fill=white, line width = \maplw] %
\tikzstyle{dashedmap}=[draw,dashed,shape=NEbox,inner sep=2pt,minimum height=\mapminh,fill=white, line width = \maplw]
\tikzstyle{mapdag}=[draw,shape=SEbox,inner sep=2pt,minimum height=\mapminh,fill=white, line width = \maplw]
\tikzstyle{mapadj}=[draw,shape=SEbox,inner sep=2pt,minimum height=\mapminh,fill=white, line width = \maplw]
\tikzstyle{maptrans}=[draw,shape=SWbox,inner sep=2pt,minimum height=\mapminh,fill=white, line width = \maplw]
\tikzstyle{mapconj}=[draw,shape=NWbox,inner sep=2pt,minimum height=\mapminh,fill=white, line width = \maplw]

\tikzstyle{medium dagmap}=[draw,shape=NEbox,inner sep=2pt,minimum height=\mapminh,fill=white,minimum width=7mm, line width = \maplw]
\tikzstyle{semilarge dagmap}=[draw,shape=NEbox,inner sep=2pt,minimum height=\mapminh,fill=white,minimum width=9.5mm, line width = \maplw]
\tikzstyle{medium dagmap adj}=[draw,shape=SEbox,inner sep=2pt,minimum height=\mapminh,fill=white,minimum width=7mm, line width = \maplw]
\tikzstyle{large dagmap}=[draw,shape=NEbox,inner sep=2pt,minimum height=\mapminh,fill=white,minimum width=12mm, line width = \maplw]

\makeatletter

\pgfdeclareshape{cornerpoint}{
\inheritsavedanchors[from=rectangle] 
\inheritanchorborder[from=rectangle]
\inheritanchor[from=rectangle]{center}
\inheritanchor[from=rectangle]{north}
\inheritanchor[from=rectangle]{south}
\inheritanchor[from=rectangle]{west}
\inheritanchor[from=rectangle]{east}
\backgroundpath{
\southwest \pgf@xa=\pgf@x \pgf@ya=\pgf@y
\northeast \pgf@xb=\pgf@x \pgf@yb=\pgf@y

\pgfmathsetmacro{\pgf@shorten@left}{\pgfkeysvalueof{/tikz/shorten left}}
\pgfmathsetmacro{\pgf@shorten@right}{\pgfkeysvalueof{/tikz/shorten right}}

\pgfpathmoveto{\pgfpoint{0.5 * (\pgf@xa + \pgf@xb)}{\pgf@ya - 5pt}}
\pgfpathlineto{\pgfpoint{\pgf@xa - 8pt + \pgf@shorten@left}{\pgf@yb - 1.5 * \pgf@shorten@left}}
\pgfpathlineto{\pgfpoint{\pgf@xa - 8pt + \pgf@shorten@left}{\pgf@yb}}
\pgfpathlineto{\pgfpoint{\pgf@xb + 8pt - \pgf@shorten@right}{\pgf@yb}}
\pgfpathlineto{\pgfpoint{\pgf@xb + 8pt - \pgf@shorten@right}{\pgf@yb - 1.5 * \pgf@shorten@right}}
\pgfpathclose
}
}

\pgfdeclareshape{cornercopoint}{
\inheritsavedanchors[from=rectangle] 
\inheritanchorborder[from=rectangle]
\inheritanchor[from=rectangle]{center}
\inheritanchor[from=rectangle]{north}
\inheritanchor[from=rectangle]{south}
\inheritanchor[from=rectangle]{west}
\inheritanchor[from=rectangle]{east}
\backgroundpath{
\southwest \pgf@xa=\pgf@x \pgf@ya=\pgf@y
\northeast \pgf@xb=\pgf@x \pgf@yb=\pgf@y

\pgfmathsetmacro{\pgf@shorten@left}{\pgfkeysvalueof{/tikz/shorten left}}
\pgfmathsetmacro{\pgf@shorten@right}{\pgfkeysvalueof{/tikz/shorten right}}

\pgfpathmoveto{\pgfpoint{0.5 * (\pgf@xa + \pgf@xb)}{\pgf@yb + 5pt}}
\pgfpathlineto{\pgfpoint{\pgf@xa - 8pt + \pgf@shorten@left}{\pgf@ya + 1.5 * \pgf@shorten@left}}
\pgfpathlineto{\pgfpoint{\pgf@xa - 8pt + \pgf@shorten@left}{\pgf@ya}}
\pgfpathlineto{\pgfpoint{\pgf@xb + 8pt - \pgf@shorten@right}{\pgf@ya}}
\pgfpathlineto{\pgfpoint{\pgf@xb + 8pt - \pgf@shorten@right}{\pgf@ya + 1.5 * \pgf@shorten@right}}
\pgfpathclose
}
}

\makeatother

\pgfkeyssetvalue{/tikz/shorten left}{0pt}
\pgfkeyssetvalue{/tikz/shorten right}{0pt}

\tikzstyle{dagpoint common}=[draw,fill=white,inner sep=1pt, line width = \maplw, minimum height = 4mm, yshift=1.2pt] 
\tikzstyle{dagpoint sc}=[shape=cornerpoint,dagpoint common]
\tikzstyle{dagpoint adjoint sc}=[shape=cornercopoint,dagpoint common]
\tikzstyle{dagpoint}=[shape=cornerpoint,shorten left=4pt,dagpoint common]
\tikzstyle{dagpointadj}=[shape=cornercopoint,shorten left=5pt,dagpoint common]
\tikzstyle{dagpointconj}=[shape=cornerpoint,shorten right=5pt,dagpoint common]
\tikzstyle{dagpointtrans}=[shape=cornercopoint,shorten right=5pt,dagpoint common]
\tikzstyle{dagpointsymm}=[shape=cornerpoint,shorten left=5pt,shorten right=5pt,dagpoint common]

\tikzstyle{widedagpoint}=[dagpoint, minimum width=1 cm, inner sep=2pt]
\tikzstyle{widedagpointadj}=[dagpointadj, minimum width=1 cm, inner sep=2pt]

\newcommand{\emptydiag}{\quad \ \input{empty-diag.tikz}}
\begin{document}

\maketitle

\begin{abstract}
Dagger compact structure is a common assumption in the study of physical process theories, but lacks a clear interpretation. Here we derive dagger compactness from more operational axioms on a category. 
We first characterise the structure in terms of a simple mapping of states to effects which we call a \emph{state dagger}, before deriving this in any category with `completely mixed' states and a form of purification, as in quantum theory. 
\end{abstract}

The processes of a general theory of physics are commonly studied in terms of \emph{symmetric monoidal categories}~\cite{abramskycoecke:categoricalsemantics}, very general mathematical structures which come with a helpful graphical calculus in which processes are represented via boxes and wires~\cite{selinger2011survey}. The composition operations in such categories have a clear operational interpretation as allowing one to run processes both in `sequence' and `parallel', and thus build more general \emph{circuit diagrams} resembling experiments~\cite{coecke2011categories}.

 When studying theories resembling quantum or classical physics, it is common to consider categories with the further structure of being \emph{dagger compact}~\cite{Selinger2007139,selinger2012finite}. Compactness is a property of a category allowing one to 'bend wires' in diagrams, often motivated by its connection to quantum teleportation~\cite{abramskycoecke:categoricalsemantics}. As well as this, the extra structure of a \emph{dagger} allows one to `flip diagrams upside-down' as below:
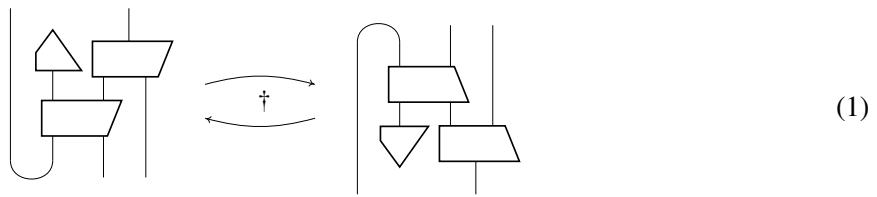
\begin{equation} \label{eq:intro-pic}
 \scalebox{0.9}{
\input{intro-pic-full.tikz}
}
\end{equation}
 Though diagrammatically and mathematically well-behaved, dagger compactness adds a significant amount of extra structure to a category, and the dagger lacks a clear meaning in terms of processes. As such, several authors have sought to understand the dagger in more operational language~\cite{westerbaan2018dagger,selby2016process,barnum2013symmetry}.
 
 In this article, we explore elementary operational axioms on a category which induce a dagger compact structure. More precisely, we give axioms on a symmetric monoidal category coming with chosen \emph{discarding} $\discard{}$ and \emph{completely mixed states} $\discardflip{}$ which ensure that it is dagger compact. 
 
 Our starting observation is that in a compact category the dagger is determined by its mapping from \emph{states} (processes with no input) to \emph{effects} (with no output). In Section~\ref{sec:state-daggers} we axiomatise when such a mapping induces a dagger compact structure on a category, calling such a map a \emph{state dagger}. In Sections~\ref{sec:extending-daggers} and~\ref{sec:CP-axiom} we show how such mapping may be extended from any suitable subcategory, in passing noting a new characterisation of Selinger's $\CPM$ construction~\cite{Selinger2007139}. 
 
 Inspiration for our approach comes from quantum theory, where the dagger has a clear operational meaning on those states which are \emph{pure}. Thanks to \emph{purification} we may then extend this assignment to arbitrary processes in the theory. Our main result in Section~\ref{sec:daggers-from-purification} applies our results to derive dagger compact structure in any category satisfying a suitable form of purification. In future work we hope to extend of our approach to include classical as well as quantum systems. 
 
  Motivation for our work comes from several recent  \emph{reconstructions} of finite-dimensional quantum theory in terms of elementary categorical axioms, each of which rely on a dagger compact structure~\cite{catreconstruction,selby2018reconstructing}.  In future work we hope to apply these results to remove such assumptions from these reconstructions, in order to provide them with a full operational interpretation. 
  
\paragraph{Related work}
Several authors have attempted to derive daggers previously.
Notably, Bas Westerbaan has derived a dagger on a `pure' subcategory from a series of axioms on an \emph{effectus}~\cite{westerbaan2018dagger}, without restricting to finite dimension or requiring monoidal structure. Selby and Coecke have characterised the dagger on pure quantum processes in terms of a mapping of states to effects (called a \emph{test structure}) within the setting of locally tomographic theories \cite{selby2016process}. Our results differ by treating dagger compactness as one entity and on the whole category, rather than just considering the dagger on pure maps.
 
Elsewhere, Barnum, Duncan and Wilce studied dagger compactness for finite-dimensional convex operational models~\cite{barnum2013symmetry}, and Cunningham and Heunen considered chosen completely mixed states in \cite{cunningham2017purity}.

\section{Setup} \label{sec:setup}

Throughout we will work in the following general setting for describing composable processes~\cite{coecke2011categories}. Recall that a category $\catC$ is \emph{symmetric monoidal} when it comes with a functor $\otimes \colon \catC \times \catC \to \catC$, a distinguished \emph{unit object} $I$ and natural \emph{coherence isomorphisms}
$(A \otimes B) \otimes C \isomto A \otimes (B \otimes C)$, $A \otimes I \isomto A$, and $\sigma \colon A \otimes B \simeq B \otimes A$, 
satisfying some straightforward equations~\cite{coecke2011categories}. Such categories are most easily treated using their \emph{graphical calculus} in which morphisms are depicted as boxes, read from bottom to top, with
\[
\input{identity-morphism.tikz}
\qquad \qquad 
\input{composition-morphism.tikz}
\qquad \qquad
\input{tensor-morphism.tikz}
\]
The (identity on) the unit object $I$ is the empty picture, so that morphisms $I \to A$, called \emph{states}, are depicted with `no input', \emph{effects} $A \to I$ have `no output', and \emph{scalars} $I \to I$ have neither. The swap isomorphism is depicted by `crossing wires' $\tinyswap$. For more on the graphical calculus see~\cite{selinger2011survey}.

The categories we will work also typically come with the following operational features allowing us to throw systems away, as well as prepare them in a `maximally noisy' state.

\begin{definition}
We say that $\catC$ has \emph{discarding} when every object $A$ comes with a chosen effect $\discard{A}$, such that
\[
\input{discard_axioms.tikz}
\]
Dually, $\catC$ has \emph{completely mixed states}~\cite{cunningham2017purity} when every object $A$ has a chosen state $\discardflip{A}$, satisfying 
\[
\input{discard_axioms_flip.tikz}
\]
We call a morphism $f \colon A \to B$ \emph{causal} when $\discard{B} \circ f= \discard{A}$ and \emph{co-causal} when $f \circ \discardflip{A} = \discardflip{B}$.
\end{definition}

\begin{remark}
We will not require completely mixed states to be normalised, so that in general $\discard{} \circ \discardflip{} \neq \id{I}$.
\end{remark}

Our categories of interest will satisfy the following property allowing one to `bend wires' in diagrams. A symmetric monoidal category is \emph{compact} when every object $A$ has a \emph{dual object}, that is an object $A^*$ along with a state $\tinycup$  of $A^* \otimes A$ (the `cup') and effect $\tinycap$ of $A \otimes A^*$ (the `cap') satisfying the \emph{snake equations}:
\[
\input{yank-shrunk.tikz}
\]
where in the right-hand sides above we draw the identity morphism on $A$ as an upward directed wire and that on $A^*$ as a downward directed one. For any dual object we define a further state and effect
\[
\input{newcup-cap.tikz}
\]
In this article we will be interested in deriving the following useful extra structure on our categories, giving the ability to `reverse' morphisms. A \emph{dagger category} is a category $\catC$ coming with an identity-on-objects contravariant involutive endofunctor $(-)^\dagger$. In other words, for each morphism $f \colon A \to B$ there is a chosen morphism $f^\dagger \colon B \to A$, satisfying
\[
f^{\dagger \dagger} = f \qquad \id{A}^\dagger = \id{A} 
\qquad (g \circ f)^\dagger = f^\dagger \circ g^\dagger
\]
A \emph{dagger symmetric monoidal} category moreover has that $(f \otimes g)^\dagger = f^\dagger \otimes g^\dagger$ and that each coherence isomorphism is a \emph{unitary}, i.e.~an isomorphism $u$ with $u^{-1} = u^\dagger$. Finally, such a category $\catC$ is \emph{dagger compact} when every object has a \emph{dagger dual}, i.e.~a dual whose cup and cap satisfy
\[
\input{dagdual.tikz} 
\ \ = \ \ 
\left(
\input{dagdual2-alt.tikz} 
\right)^\dagger
\]
When $\catC$ has discarding it automatically has a choice of completely mixed states 
$\discardflip{} = \discard{}^\dagger$ 
and we additionally require on each object that 
\begin{equation} \label{eq:duals-discard}
\input{ax-1-2-alt.tikz}
\qquad
\qquad
\input{ax-1-1-alt.tikz}
\end{equation}

Dagger categories come with their own graphical calculus~\cite{selinger2011survey} in which the dagger corresponds to reflecting morphisms vertically as in~\eqref{eq:intro-pic}. However here we will be deriving the existence of daggers in the first place, and so not use such diagrammatic rules.


\begin{examples}
Our main examples of dagger compact categories with discarding are the following.
\begin{enumerate}[leftmargin=*]
\item 
In $\FCStar$ the objects are finite-dimensional C*-algebras and the morphisms completely positive maps $f \colon A \to B$. Here $I=\mathbb{C}$ and $\otimes$ is the standard one of such algebras, with the scalars given by $\mathbb{R}^{\geq 0}$. The dagger is given by the Hermitian adjoint, $\discard{}$ by the assignment $a \mapsto \Tr(a)$, and $\discardflip{A}$ is the assignment $1 \mapsto 1_A$. 
The commutative algebras give a subcategory $\Class$.

\item 
$\Quant{}$ is defined as the full subcategory of $\FCStar$ on algebras of the form $B(\hilbH)$. In particular states correspond to unnormalised density matrices, with $\discardflip{}$ given by the identity matrix. The 'cup' and 'cap' are given by the \emph{Bell state and effect}, respectively.

 
\item  More generally, we may consider $\Quant{S} := \CPM(\MatS)$ for any \emph{phased field}~\cite{catreconstruction}, i.e.~any involutive field  $(S, \dagger)$ whose elements of the form $a^\dagger \cdot a$ are closed under addition. Then $\Quant{\mathbb{C}} \simeq \Quant{}$ while $\Quant{\mathbb{R}}$ is quantum theory over \emph{real Hilbert spaces}~\cite{Hardy2012Holism}.


\item 
In $\Rel$, objects are sets and morphisms $R \colon A \to B$ are relations $R \subseteq A \times B$. Composition is the usual one of relations, with $I=\{\star\}$ the singleton set, $\otimes$ given by the Cartesian product, and the dagger by relational converse. $\discard{A}$ is the relation $A \to I$ with $a \mapsto \star$ for all $a \in A$, with $A^*=A$ and $\tinycup$ picking out the elements of the form $(a,a)$. This example may be generalised to $\Rel(\catC)$ for any \emph{regular category} $\catC$, or indeed any \emph{bicategory of relations}~\cite{CatsOfRelations,carboni1987cartesian}.


\end{enumerate}
\end{examples}

\section{State daggers} \label{sec:state-daggers}
Though the dagger often lacks a clear interpretation, we will see that it often has a clearer meaning on (certain) states. We now reduce dagger compactness to the following more lightweight structure.

\begin{definition}
A symmetric monoidal category $\catC$ has a \emph{state dagger} when each object $A$ comes with a mapping of states to effects
\[
\input{statemap.tikz}
\]
such that $\peassign{\id{I}}=\id{I}$ and for all states $\psi, \phi$ we have
\begin{equation} \label{eq:compositional} 
\peassign{
\left(\begin{tikzpicture}
	\begin{pgfonlayer}{nodelayer}
		\node [style=point] (0) at (-5, -0.5) {$\psi$};
		\node [style=none] (1) at (-5, 1) {};
		\node [style=point] (7) at (-3.25, -0.5) {$\phi$};
		\node [style=none] (8) at (-3.25, 1) {};
	\end{pgfonlayer}
	\begin{pgfonlayer}{edgelayer}
		\draw (0) to (1.center);
		\draw (7) to (8.center);
	\end{pgfonlayer}
\end{tikzpicture}
 \right)
} 
\ \  = \ \
\input{statemap-ax-1.tikz}
\qquad \qquad  \qquad
\peassign{
\left(\input{ax-3.tikz}
\right)
}
\ \ = \ \ 
\input{ax-3a.tikz}
\\ 
\end{equation}
and for all  states $\psi$ and coherence isomorphisms  $\gamma (= \alpha, \rho, \lambda, \sigma)$ of $\catC$ we have
\begin{equation}
\label{enum:comp-state-coherence}
\peassign{
\left(\input{state-coherence.tikz}\right)
}
 \ \ = \ \ 
\input{state-coherence-2.tikz}
\end{equation}
We say an object has a \emph{state dagger dual} when it has a dual object whose cup and cap satisfy
\begin{equation} \label{eq:state-dag-dual}
\input{dagdual.tikz} 
\ \ = \ \ 
\peassign{
\left(
\input{dagdual2-alt.tikz} 
\right)}
\end{equation}
\end{definition}

\begin{proposition} \label{prop:state-dags}
Specifying a dagger compact structure on a symmetric monoidal category $\catC$ is equivalent to specifying a state dagger for which every object has a state dagger dual.
\end{proposition}
\begin{proof}
If $\catC$ is dagger compact then the assignment $\peassign{\psi} = \psi^\dagger$ yields a state dagger since the dagger  is a monoidal functor and all coherence isomorphisms are unitary. Then~\eqref{eq:state-dag-dual} is simply the statement that every object has a dagger dual. Conversely, let us suppose that $\catC$ comes with such a state dagger.

We first show that the state dagger is injective. Suppose that $\peassign{\psi} = \peassign{\phi}$ for two states $\psi, \phi$. Then choosing any state dagger dual of their domain we have
\[
\input{injec-1-alt.tikz}
\ \ = \ \ 
\peassign{\left(\input{injec-2-alt.tikz}\right)}
\ \ = \ \ 
\peassign{\left(\input{injec-3-alt.tikz}\right)}
\ \ = \ \ 
\input{injec-4-alt.tikz}
\]
Then composing with $\tinycup$ on the left shows that $\psi = \phi$. Next, note that applying~\eqref{eq:compositional} and \eqref{enum:comp-state-coherence} we obtain the rule
\begin{equation} \label{eq:coherence-reduce}
\peassign{\left(
\input{state-coherence-reduce.tikz}
\right)}
\ \ = \ \ 
\peassign{\left(
\input{scr1.tikz}
\right)}
\ \ = \ \ 
\input{scr2.tikz}
\ \ = \ \ 
\input{state-coherence-reduce-2.tikz}
\end{equation}
for all states $\psi, \phi$. Now let $f \colon A \to B$ be any morphism and choose a state dagger dual $B^*$ for $B$. We define $f^\dagger$ to be the unique morphism such that 
\begin{equation} \label{eq:def-dagg}
\peassign{\left(
\input{unique-1-alt.tikz}
\right)}
\ \ = \ \ 
\input{unique-4-alt.tikz}
\end{equation}
Indeed since the state dagger is injective, any such morphism must be unique. To show that it exists, choose a state dagger dual $A^*$ for $A$ and then note that
 \[
\input{new-dag-0a.tikz} \ \ 
:= \ \ 
\peassign{\left(\input{new-dag-0b-alt.tikz}\right)}
\implies
\peassign{\left(
\input{unique-2-alt.tikz}
\right)}
\ \ = \ \ 
\input{unique-3-alt.tikz}
\ \ = \ \ 
\input{unique-4-alt.tikz}
 \]
as required, using~\eqref{eq:coherence-reduce} in the first equality and the snake equations in the second. Since $f'$ does not depend on our choice of dual for $B$, $f^\dagger$ does not either, and satisfies~\eqref{eq:def-dagg} for any such dual $B^*$. Now let $g \colon A' \to B'$ be any other morphism, and choose a state dagger dual $B'^*$ for $B'$. One may straightforwardly check that the state $\omega$ defined by 
\[
\input{tensor-dual-alt.tikz}
\]
exhibits $(B^* \otimes B'^*)$ as a state dagger dual of $(B \otimes B')$.
Then we have that $(f \otimes g)^\dagger = f^\dagger \otimes g^\dagger$ since 
\[
\peassign{\left(
\input{tensor-1-alt.tikz}
\right)}
\ \ = \ \ 
\peassign{\left(
\input{tensor-2-alt.tikz}
\right)}
\ \ = \ \ 
\input{tensor-3-alt.tikz}
\ \ = \ \ 
\input{tensor-4-alt.tikz}
\]
Next consider any morphism $g \colon B \to C$, and choose a state dagger dual $C^*$ for $C$. Then we have
\[
\peassign{\left(
\input{compose-1-alt.tikz}
\right)}
\ \ = \ \ 
\peassign{\left(
\input{compose-2-alt.tikz}
\right)}
\ \ = \ \ 
\input{compose-3-alt.tikz}
\ \ = \ \ 
\input{compose-4-alt.tikz}
\]
Hence $(g \circ f)^\dagger = f^\dagger \circ g^\dagger$. By definition $\id{A}^\dagger = \id{A}$ for all objects $A$. The fact that every coherence isomorphism is unitary follows from~\eqref{enum:comp-state-coherence} when $\psi$ is the state of a state dagger dual. This makes the dagger a strict symmetric monoidal functor. Moreover since by the second equation of~\eqref{eq:compositional} we have $\psi^\dagger = \peassign{\psi}$ for all states $\psi$. 

To check that the dagger is involutive, by bending wires one may see that it suffices to show that $\psi^{\dagger\dagger}=\psi$ for all states $\psi$. Now setting $f=\psi^\dagger$ and choosing the trivial state dagger dual $I^* = I$,~\eqref{eq:def-dagg} tells us that by definition $\phi := \psi^{\dagger\dagger}$ is the unique state $\phi$ with $\peassign{\phi} = \psi^\dagger$. But since $\peassign{\psi} = \psi^\dagger$ we have $\phi=\psi$ by injectivity of the state dagger. Every state dagger dual is then a dagger dual, making $\catC$ dagger compact. 

Finally note that this is a correspondence between dagger compact structures and such state daggers. Indeed we have seen that $\peassign{\psi}=\psi^\dagger$ for all states $\psi$. Conversely, when defining $\peassign{(-)}$ in this way, every dagger satisfies our definition of $f^\dagger$ above since it is an involutive monoidal functor. 
\end{proof}

This result extends to include discarding morphisms as follows.

\begin{proposition} \label{prop:disc-dag-comp}
Let $(\catC, \discardflip{})$ be a symmetric monoidal category with completely mixed states. Specifying a compatible dagger compact structure on $\catC$ is equivalent to specifying a state dagger such that each object has a state dagger dual satisfying~\eqref{eq:duals-discard}, where $\discard{} := \peassign{(\discardflip{})}$.
\end{proposition}
\begin{proof}
By Proposition~\ref{prop:state-dags}, noting that the condition is equivalent to compatibility of the duals with $\discard{}$.
\end{proof}

\section{Daggers from dilations} \label{sec:extending-daggers}

In defining a dagger compact structure on a whole category, it can sometimes be useful to extend an existing one defined on a subcategory, as we will see for quantum theory in the next section. First, we call a state $\sigma$ a \emph{dilation} of a state $\rho$ when
\[
\input{dilation-state.tikz}
\]
By a \emph{dilation structure} we mean a symmetric monoidal category $\catC$ with discarding $\discard{}$ coming with a chosen symmetric monoidal subcategory $\dsubcat$ such that every state in $\catC$ has a dilation belonging to $\catD$. We call a dilation structure \emph{dagger compact} when $\catC$ is a dagger compact category with discarding and $\dsubcat$ is a dagger compact subcategory of $\catC$. 

\begin{theorem} \label{thm:dag-compact-from-dilation}
A dilation structure $(\catC,\dsubcat,\discard{})$ is dagger compact iff $\catC$ has chosen completely mixed states, $\dsubcat$ has a state dagger satisfying
\begin{equation} \label{eq:dag-resp-state}
\input{ax-5.tikz}
\end{equation}
and in $\catD$ every object $A$ has a state dagger dual $A^*$ which in $\catC$ satisfies \eqref{eq:duals-discard}.
\end{theorem}
\begin{proof}
By Proposition~\ref{prop:state-dags} the state dagger makes $\dsubcat$ dagger compact; we now extend it to $\catC$ as follows. For any state $\rho$, let $\psi$ be a dilation of $\rho$ belonging to $\dsubcat$. Then we set
\begin{equation} \label{eq:ext-dag-def-alt}
\input{state-dag-def.tikz}
\end{equation}
Thanks to \eqref{eq:dag-resp-state} this is independent of our choice of dilation $\psi$. Since any state in $\dsubcat$ of an object $A$ may be seen as a dilation of itself (via the object $I$ and coherence isomorphism $A \otimes I \simeq{} I$), this coincides with the state dagger on $\dsubcat$. Let us check that this is compositional. 
Since dilations compose under $\otimes$, it is easy to see that the left-hand equation of~\eqref{eq:compositional} is satisfied. Since all coherence isomorphisms belong to $\dsubcat$ the condition~\eqref{enum:comp-state-coherence} also lifts easily to $\catC$. By definition $\peassign{(\discardflip{})} = \discard{}$ since $\discardflip{}$ has dilation $\tinycup$ and 
\[
\input{bendy-arg.tikz}
\]
Let $\rho, \sigma \in \catC$ with respective dilations $\psi, \phi \in \dsubcat$. Then using our assumption on duals and the definition of our state dagger we have 
\[
\peassign{
\left(\input{ax-3-again.tikz}
\right)
}
\ \ = \ \ 
\peassign{
\left(\input{ax-3-again-3a.tikz}
\right)
}
\ \ = \ \ 
\input{ax-3-again-3a-flip.tikz}
\ \ = \ \ 
\input{ax-3a-again.tikz}
\]
Next note that since every object of $\catC$ has a state (e.g. $\discard{}$), $\catC$ and $\dsubcat$ have the same objects by the definition of a dilation structure. Hence every object has a dual. Finally, by assumption on $\dsubcat$ these duals satisfy the requirements of Proposition~\ref{prop:disc-dag-comp}.
\end{proof}

\section{Daggers in $\CPM$ categories} \label{sec:CP-axiom}

A useful source of dagger compact categories with discarding are those $\CPM(\dsubcat)$ arising from Selinger's \emph{CPM construction} on a dagger compact category $\dsubcat$~\cite{Selinger2007139}. The motivating example is
\[
\Quant{} \simeq \CPM({\Quant{}}_\pure)
\]
where ${\Quant{}}_\pure$ is the subcategory of `pure' morphisms, explored in the next section. Coecke has shown that such categories correspond precisely, via $\catC \simeq \CPM(\dsubcat)$, to dagger compact dilation structures $(\catC, \dsubcat, \discard{})$ satisfying the \emph{CP axiom}~\cite{coecke2008axiomatic}:
\begin{equation*} 
\input{CP.tikz}
\end{equation*}
for all $f, g \in \dsubcat$ with the same domain. This indeed holds for all pure morphisms $f, g$ in $\Quant{}$. 

The analogous condition for state daggers allows us to simplify our earlier results by dropping several conditions. We can characterise $\CPM$ categories in terms of state daggers as follows.

\begin{theorem} \label{thm:CP-axiom}
A category $\catC$ arises from the $\CPM$ construction precisely when it belongs to a dilation structure ($\catC$, $\dsubcat$, $\discard{}$) with completely mixed states $\discardflip{}$ and a state dagger on $\dsubcat$ satisfying 
\begin{equation} \label{eq:CP-state-dagger}
\input{ax-4.tikz}
\end{equation}
for all $\psi, \phi \in \dsubcat$, and such that in $\dsubcat$ every object has a state dagger dual satisfying \eqref{eq:duals-discard}. Moreover, we no longer require condition~\eqref{enum:comp-state-coherence} of a state dagger but only the special case 
\begin{equation} \label{eq:symmetry-statemap}
\peassign{\left(
\input{symmetry-statemap.tikz}
\right)}
\ \ = \ \ 
\input{symmetry-statemap-2.tikz}
\end{equation}
for all states $\psi$ in $\dsubcat$.
\end{theorem}
\begin{proof}
For any dagger compact dilation structure satisfying the CP axiom, as before set $\peassign{\psi} = \psi^\dagger$ for all states $\psi$. By bending wires, ~\eqref{eq:CP-state-dagger} is equivalent to the CP axiom, and~\eqref{eq:symmetry-statemap} holds since the dagger is symmetric monoidal.

Conversely, let us show that $\dsubcat$ satisfies the requirements of Proposition~\ref{prop:state-dags}. Inspecting the proof, we see that since the rule \eqref{eq:symmetry-statemap} allows us to deduce~\eqref{eq:coherence-reduce}, we only require the condition~\eqref{enum:comp-state-coherence} on coherence isomorphisms $\gamma$ in the case where $\psi$ is the state of a state dagger dual (in order to deduce that $\gamma$ is unitary). Let us establish this more generally for any causal isomorphism $\gamma$ in $\dsubcat$. We need to show that the effect $\gamma'$ defined by 
\begin{equation} \label{eq:coher-simple}
\input{gamma.tikz}
\ \ := \ \ 
\peassign{\left(\input{state-coherence-2-alt.tikz}\right)}
\qquad \qquad 
\text{ satisfies }
\qquad \qquad
\input{state-coherence-simpler-4-alt.tikz}
\end{equation}
Now any coherence isomorphism belongs to $\dsubcat$ by assumption, and is causal in any category with discarding~\cite[Chapter~1]{thesis}. Hence by~\eqref{eq:CP-state-dagger} we have implications 
\[
\input{state-coherence-simpler-alt.tikz}
\quad
\text{ and so }
\quad
\input{state-coherence-simpler-3-alt.tikz}
\]
by~\eqref{eq:CP-state-dagger}, from which~\eqref{eq:coher-simple} follows as required. 
Hence by Proposition~\ref{prop:state-dags} $\dsubcat$ is dagger compact.

By Theorem~\ref{thm:dag-compact-from-dilation} it remains to check that~\eqref{eq:dag-resp-state} holds in $\catC$. Suppose that $\psi, \phi \in \dsubcat$ satisfy the left-hand side of~\eqref{eq:dag-resp-state}. Then we must show that the effect
\[
\input{psi-eff-alt.tikz}
\]
is equal to the effect obtained by replacing $\psi$ by $\phi$.
But by \eqref{eq:CP-state-dagger} this effect is determined by the morphism
\[
\input{psi-cp-alt-2.tikz}
\]
using naturality of $ \ \tinyswap$ in the second step. But by (the horizontal reflection of) \eqref{eq:CP-state-dagger} the latter morphism is determined by applying $\discard{}$ to the right output of $\psi$, which is equal to doing the same to $\phi$.
\end{proof}

\section{Daggers from purification} \label{sec:daggers-from-purification}

Let us now apply our earlier results to derive dagger compactness from operational axioms on a category $\catC$ with discarding and completely mixed states.


First, recall that $\catC$ has \emph{zero morphisms} when there is a morphism $0 \colon A \to B$ between any two objects $A, B$ which compose via $\circ $ and $\otimes$ with any morphism to give $0$. We then say that $\catC$ satisfies \emph{normalisation} when every non-zero state is of the form
\[
\input{normalisation.tikz}
\]
for some scalar $r$ and unique causal state $\sigma$, which we call the \emph{normalisation} of $\rho$. Normalisation is automatic when all non-zero scalars are invertible, as in our main examples here. Our axioms will concern morphisms of the following form. 

\begin{definition}~\cite{chiribella2014distinguishability} A morphism $f$ is \emph{pure} when either $f = 0$ or $f$ satisfies
\begin{equation} \label{eq:wholedef}
\input{whole-def-no-labels.tikz}
\ \ \ 
\text{  for some causal $\rho$}
\end{equation}
 This notion of purity in fact typically coincides with the more usual one in terms of mixing~\cite[Chap.~4]{thesis}. Indeed, in $\Quant{}$ a morphism is pure in this sense precisely when it is given by a Kraus operator.
\end{definition}

\begin{lemma} \label{lem:normalisation-consequences}
Let $\catC$ satisfy normalisation. Then pure states are closed under normalisation and $\otimes$.
\end{lemma}
\begin{proof}

For the first part, say $\psi = \phi \circ r$ for a scalar $r$, where $\phi$ is causal. Then if some state $\rho$ dilates $\phi$ it is causal, and moreover $\rho \circ r$ dilates $\psi$. Hence $\rho \circ r = \psi \otimes \sigma = (\phi \otimes \sigma) \circ r$ for some causal state $\sigma$. From uniquness of normalisation and causality of these states we obtain $\rho = \phi \otimes \sigma$, as required. 

Next, suppose that states $\psi, \phi$ are pure, and $\rho$ is some dilation of $\psi \otimes \phi$. Then the normalisation of $\rho$ is also a dilation of the tensor $\psi' \otimes \phi'$, of the respective normalisations of $\psi, \phi$. Hence by the previous part it suffices to consider when $\psi$ and $\phi$ are causal. 
Then we have implications
\[
\input{norm-arg-0.tikz}
\implies
\input{norm-arg-1.tikz}
\]
Hence $\rho = \psi \otimes \sigma$ for some causal state $\sigma$. Composing with $\discard{}$ on the left then shows that $\sigma$ dilates $\phi$ and so $\sigma = \phi \otimes \tau$ for some causal state $\tau$. Hence we have $\rho = \psi \otimes \phi \otimes \tau$, as required.
\end{proof}

\begin{definition}
We consider symmetric monoidal categories $(\catC, \discard{}, \discardflip{})$ with discarding, completely mixed states, and normalisation, satisfying the following axioms. 
\begin{enumerate}[leftmargin=*]
\item \paxiom{Purification}~\cite{chiribella2010purification}\label{ax:purification}
 Every state $\rho$ has a \emph{purification}, i.e.~a dilation
\[
\input{purif.tikz}
\]
for which $\psi$ is pure. Moreover, purifications are \emph{essentially unique} in that 
\begin{equation} \label{eq:EU}
\input{EU-purif-states.tikz}
\end{equation}
for some causal and co-causal isomorphism $U$ on $B$, for all such pure states $\psi, \phi$. 
\item \paxiom{Sharpness}~\cite{Hardy2011a}
For every causal pure state $\psi$ there is a unique pure co-causal effect $\peassign{\psi}$ with
\begin{equation} \label{eq:state-dagger-as-property}
\begin{tikzpicture}
	\begin{pgfonlayer}{nodelayer}
		\node [style=point] (0) at (0, -1) {$\psi$};
		\node [style=copoint] (1) at (0, 1.25) {$\peassign{\psi}$};
		\node [style=none] (2) at (2, 0) {$=$};
	\end{pgfonlayer}
	\begin{pgfonlayer}{edgelayer}
		\draw (0) to (1);
	\end{pgfonlayer}
\end{tikzpicture}

\emptydiag
\end{equation}
Moreover, $\psi$ is conversely the unique causal pure state satisfying the above.

We may then similarly assign an effect to each non-zero (not necessarily causal) pure state $\psi$ via 
\begin{equation} \label{eq:extended-state-dagger}
\input{pe-extend.tikz}
\end{equation}
where $\phi$ is the (pure by Lemma~\ref{lem:normalisation-consequences}) normalisation of $\psi$, and $r=\discard{} \circ \psi$. We also set $\peassign{0}= 0$.

\item \paxiom{Pure composition} \label{ax:purecomposition}
For all causal pure states $\psi, \phi$, the following state and effect are pure with
\begin{equation} \label{eq:pure-composition}
\peassign{
\left(\input{ax-3.tikz}
\right)
}
\ \ = \ \ 
\input{ax-3a.tikz}
\end{equation}
\item \paxiom{Pre-duals}  \label{ax:unifpurif}
Each completely mixed state $\discardflip{A}$ has a purification $\omega$ satisfying
\begin{equation} \label{eq:pre-dual}
\input{axiom-b.tikz} 
\qquad \quad 
\input{axiom-b-3.tikz} 
\qquad \quad 
\input{axiom-b-2.tikz} 
\end{equation}
\item \paxiom{Identity tomography} \label{ax:Id-Tomog}
For all endomorphisms $V$ we have 
\[
\left( \ \ 
\input{isomorphism-tomography.tikz}
\ \ 
\forall \text{ pure } \psi 
\right)
\ \implies \ 
\input{isomorphism-tomography-2.tikz}
\]
\end{enumerate}
\end{definition}

\begin{examples}  \label{ex:pure-properties-example}
$\Quant{S}$ satisfies these axioms, for any phased field $S$, 
including $\Quant{}$ and $\Quant{\mathbb{R}}$. More generally, they follow from the `operational principles' of~\cite{catreconstruction}; as we prove in appendix~\ref{sec:CP-axiom}.

All of the axioms aside from purification and purity of $\omega$ in the pre-duals axiom hold in $\FCStar$ and $\Rel$. In $\FCStar$ pure morphisms between direct sums of quantum algebras are simply those induced by some single pure morphism $B(\hilbH_i) \to B(\hilbH_j)$ in $\Quant{}$. In $\Rel$ a morphism is pure iff it is empty or a singleton. The failure of purification in these examples has led to alternative definitions of purity~\cite{selby2018reconstructing,cunningham2017purity}.
\end{examples}

Let us discuss these axioms in more detail. Sharpness essentially appears as the first axiom in Hardy's reconstruction of quantum theory~\cite{Hardy2011a} and is crucial in allowing us to define a (state) dagger as a \emph{property} rather than extra structure. The characterisation of  $\peassign{\psi}$ is also similar to the `test structures' considered in~\cite{selby2016process} and relates to causal pure states being \emph{dagger kernels}; see appendix~\ref{sec:deriving-axioms}.

 Purification is a standard axiom for capturing quantum theory~\cite{chiribella2010purification,PhysRevA.84.012311InfoDerivQT,catreconstruction}. We require each $U$ to be co-causal since in $\Quant{}$ each $\discardflip{}$ is the unique state (up to scalars) preserved by all causal isomorphisms~\cite{PhysRevA.84.012311InfoDerivQT}. Symmetric monoidality of $\catC_\pure$ mainly involves pure morphisms being closed under $\circ$ by Lemma~\ref{lem:normalisation-consequences}. In the compact setting it suffices that each state~\eqref{eq:pure-composition} is pure (axiom 5 of~\cite{PhysRevA.84.012311InfoDerivQT}). The pre-duals axiom and~\eqref{eq:pure-composition} are necessary for dagger compactness. Finally, identity tomography is a special case of more general tomography axioms~\cite{Barrett2007InfoGPTs,chiribella2010purification}. It is weaker than \emph{local tomography}, since this fails in $\Quant{\mathbb{R}}$~\cite{Hardy2012Holism,PhysRevA.84.012311InfoDerivQT}.




Our main result is now the following. 

\begin{theorem}
Let $\catC$ be a symmetric monoidal category with discarding $\discard{}$ and completely mixed states $\discardflip{}$, satisfying normalisation and the above axioms. Then $\catC$ forms a dagger compact category with discarding, with the pure morphisms $\catC_\pure$ as a dagger compact subcategory.
\end{theorem}

\begin{proof}
We will show that $\catC_\pure$ is a subcategory with state dagger~\eqref{eq:extended-state-dagger}. Let us first show that $\catC$ has duals whose cups are pure and satisfy~\eqref{eq:state-dag-dual}, \eqref{eq:duals-discard}. First note that by definition for all pure states $\psi$ we have
\begin{equation} \label{eq:ex-state-dagger-dual}
\input{ext-dag-rule-1.tikz}
\end{equation}
Now for any object $A$, let $\omega$ be a pure state of $A \otimes B$ as in~\eqref{eq:pre-dual}, and define
\[
\input{snake-map.tikz}
\]
Then for all causal pure states $\psi$ of $A$ by~\eqref{eq:pure-composition} we have
\[
\input{snake-map-psi-alt.tikz} \emptydiag
\]
But $V$ is causal by construction and hence so is $V \circ \psi$. Moreover this state is pure by the assumption that all states and effects as in~\eqref{eq:pure-composition} are. Hence we have $V \circ \psi = \psi$. By normalisation the same holds for arbitrary pure states $\psi$. Hence by identity tomography we have $V = \id{}$. Hence $\omega$ satisfies the first snake equation, and the other equation holds similarly. By construction these duals satisfy~\eqref{eq:state-dag-dual}, \eqref{eq:duals-discard}. 

 Since $\catC$ has pure cups and caps, all isomorphisms in $\catC$ are pure, including the coherence ones. Moreover bending wires in~\eqref{eq:pure-composition} now shows that pure morphisms are closed under composition $\circ$, and they are closed under $\otimes$ by Lemma~\ref{lem:normalisation-consequences}. Hence $\catC_\pure$ is a symmetric monoidal subcategory of $\catC$.

Now let us show that~\eqref{eq:extended-state-dagger} indeed defines a state dagger on $\catC_\pure$. Firstly, note that for arbitrary (not necessarily causal) pure states $\psi$, the effect $\peassign{\psi}$ is pure since all scalars are by Lemma~\ref{lem:normalisation-consequences}. By definition we have $\peassign{\id{I}} = \id{I}$. The second condition of~\eqref{eq:compositional} holds by~\eqref{eq:pure-composition}. By normalisation it suffices to verify the remaining conditions for pure states which are causal. Now the first condition of~\eqref{enum:comp-state-coherence} follows in $\catC_\pure$ since for all causal pure states $\psi, \phi$ we have
\[
\input{derive-compact.tikz} \emptydiag
\]
Next, note that all coherence isomorphisms $\gamma$ in $\catC_\pure$ are causal and co-causal in $\catC$~\cite[Chap.~1]{thesis}. Hence~\ref{enum:comp-state-coherence} follows for any causal pure state $\psi$ since we have 
\[
 \input{state-coherence-proof.tikz} \emptydiag
\]
This establishes the presence of the state dagger. Next let us verify~\eqref{eq:dag-resp-state} for all pure states $\psi, \phi$ of some object $A \otimes B$. By normalisation it suffices to consider when $\psi, \phi$ are causal. Suppose the left-hand side of~\eqref{eq:dag-resp-state} is satisfied. Then by essential uniqueness there is some causal and co-causal isomorphism $U$ on $B$ for which the following all hold
\[
\input{Use-EU-2.tikz}
\quad
\implies
\quad
\input{Use-EU-new.tikz} \emptydiag
\quad
 \qquad 
\implies
\quad 
\input{Use-EU-3.tikz}
\]
Then since $U$ is co-causal we obtain the right-hand side of~\eqref{eq:dag-resp-state}, and so we are done by Theorem~\ref{thm:dag-compact-from-dilation}. 
\end{proof}

In principle this result should now allow us to adapt the categorical quantum reconstruction~\cite{catreconstruction} to not explicitly require a dagger, and thus be more fully operational in nature. For this one should derive our less operationally clear  `pure composition' and `pre-duals' axioms from more well-motivated ones, as we begin to explore in appendix~\ref{sec:deriving-axioms}.

\begin{remark}
 In future it would be desirable to derive dagger compactness in categories without such purification, such as $\FCStar$ and $\Rel$, for example allowing for a dagger-free version of the reconstruction~\cite{selby2018reconstructing}. Speculatively, it may suffice to characterise those \emph{self-adjoint} morphisms $U$, with $U=U^\dagger$. Indeed in both categories every state $\rho$ may be written as $U \circ \discardflip{}$ for some such morphism, so that $\rho^\dagger = \discard{} \circ U$.
\end{remark}

\bibliographystyle{eptcs}
\bibliography{thesis-bib}

\begin{thebibliography}{10}
\providecommand{\bibitemdeclare}[2]{}
\providecommand{\surnamestart}{}
\providecommand{\surnameend}{}
\providecommand{\urlprefix}{Available at }
\providecommand{\url}[1]{\texttt{#1}}
\providecommand{\href}[2]{\texttt{#2}}
\providecommand{\urlalt}[2]{\href{#1}{#2}}
\providecommand{\doi}[1]{doi:\urlalt{http://dx.doi.org/#1}{#1}}
\providecommand{\bibinfo}[2]{#2}

\bibitemdeclare{inproceedings}{abramskycoecke:categoricalsemantics}
\bibitem{abramskycoecke:categoricalsemantics}
\bibinfo{author}{S.~\surnamestart Abramsky\surnameend} \&
  \bibinfo{author}{B.~\surnamestart Coecke\surnameend} (\bibinfo{year}{2004}):
  \emph{\bibinfo{title}{{A categorical semantics of quantum protocols}}}.
\newblock In: {\sl \bibinfo{booktitle}{Logic in Computer Science 19}},
  \bibinfo{publisher}{IEEE Computer Society}, pp. \bibinfo{pages}{415--425},
  \doi{10.1109/lics.2004.1319636}.

\bibitemdeclare{article}{barnum2013symmetry}
\bibitem{barnum2013symmetry}
\bibinfo{author}{Howard \surnamestart Barnum\surnameend}, \bibinfo{author}{Ross
  \surnamestart Duncan\surnameend} \& \bibinfo{author}{Alexander \surnamestart
  Wilce\surnameend} (\bibinfo{year}{2013}): \emph{\bibinfo{title}{Symmetry,
  compact closure and dagger compactness for categories of convex operational
  models}}.
\newblock {\sl \bibinfo{journal}{Journal of philosophical logic}}
  \bibinfo{volume}{42}(\bibinfo{number}{3}), pp. \bibinfo{pages}{501--523},
  \doi{10.1007/s10992-013-9280-8}.

\bibitemdeclare{article}{Barrett2007InfoGPTs}
\bibitem{Barrett2007InfoGPTs}
\bibinfo{author}{J.~\surnamestart Barrett\surnameend} (\bibinfo{year}{2007}):
  \emph{\bibinfo{title}{{Information processing in generalized probabilistic
  theories}}}.
\newblock {\sl \bibinfo{journal}{Physical Review A - Atomic, Molecular, and
  Optical Physics}} \bibinfo{volume}{75}(\bibinfo{number}{3}),
  \doi{10.1103/PhysRevA.75.032304}.

\bibitemdeclare{article}{carboni1987cartesian}
\bibitem{carboni1987cartesian}
\bibinfo{author}{A.~\surnamestart Carboni\surnameend} \&
  \bibinfo{author}{R.~\surnamestart Walters\surnameend} (\bibinfo{year}{1987}):
  \emph{\bibinfo{title}{Cartesian bicategories {I}}}.
\newblock {\sl \bibinfo{journal}{Journal of pure and applied algebra}}
  \bibinfo{volume}{49}(\bibinfo{number}{1-2}), pp. \bibinfo{pages}{11--32},
  \doi{10.1016/0022-4049(87)90121-6}.

\bibitemdeclare{article}{chiribella2014distinguishability}
\bibitem{chiribella2014distinguishability}
\bibinfo{author}{G.~\surnamestart Chiribella\surnameend}
  (\bibinfo{year}{2014}): \emph{\bibinfo{title}{Distinguishability and
  copiability of programs in general process theories}}.
\newblock \bibinfo{note}{\bibxiv{arXiv:1411.3035}}.

\bibitemdeclare{article}{chiribella2010purification}
\bibitem{chiribella2010purification}
\bibinfo{author}{G.~\surnamestart Chiribella\surnameend},
  \bibinfo{author}{G.~M. \surnamestart D'Ariano\surnameend} \&
  \bibinfo{author}{P.~\surnamestart Perinotti\surnameend}
  (\bibinfo{year}{2010}): \emph{\bibinfo{title}{{Probabilistic theories with
  purification}}}.
\newblock {\sl \bibinfo{journal}{Physical Review A}}
  \bibinfo{volume}{81}(\bibinfo{number}{6}), p. \bibinfo{pages}{62348},
  \doi{10.1103/physreva.81.062348}.

\bibitemdeclare{article}{PhysRevA.84.012311InfoDerivQT}
\bibitem{PhysRevA.84.012311InfoDerivQT}
\bibinfo{author}{G.~\surnamestart Chiribella\surnameend},
  \bibinfo{author}{G.~M. \surnamestart D'Ariano\surnameend} \&
  \bibinfo{author}{P.~\surnamestart Perinotti\surnameend}
  (\bibinfo{year}{2011}): \emph{\bibinfo{title}{{Informational derivation of
  quantum theory}}}.
\newblock {\sl \bibinfo{journal}{Phys. Rev. A}}
  \bibinfo{volume}{84}(\bibinfo{number}{1}), p. \bibinfo{pages}{12311},
  \doi{10.1103/PhysRevA.84.012311}.

\bibitemdeclare{article}{coecke2008axiomatic}
\bibitem{coecke2008axiomatic}
\bibinfo{author}{B.~\surnamestart Coecke\surnameend} (\bibinfo{year}{2008}):
  \emph{\bibinfo{title}{Axiomatic description of mixed states from {S}elinger's
  {CPM}-construction}}.
\newblock {\sl \bibinfo{journal}{Electronic Notes in Theoretical Computer
  Science}} \bibinfo{volume}{210}, pp. \bibinfo{pages}{3--13},
  \doi{10.1016/j.entcs.2008.04.014}.

\bibitemdeclare{incollection}{coecke2011categories}
\bibitem{coecke2011categories}
\bibinfo{author}{B.~\surnamestart Coecke\surnameend} \&
  \bibinfo{author}{E.~\surnamestart Paquette\surnameend}
  (\bibinfo{year}{2011}): \emph{\bibinfo{title}{{Categories for the practising
  physicist}}}.
\newblock In: {\sl \bibinfo{booktitle}{New Structures for Physics}},
  \bibinfo{publisher}{Springer Berlin Heidelberg}, pp.
  \bibinfo{pages}{173--286}, \doi{\detokenize{10.1007/978-3-642-12821-9_3}}.

\bibitemdeclare{inproceedings}{cunningham2017purity}
\bibitem{cunningham2017purity}
\bibinfo{author}{O.~\surnamestart Cunningham\surnameend} \&
  \bibinfo{author}{C.~\surnamestart Heunen\surnameend} (\bibinfo{year}{2018}):
  \emph{\bibinfo{title}{Purity through Factorisation}}.
\newblock In: {\sl \bibinfo{booktitle}{{Proceedings of the 14th International
  Conference on} Quantum Physics and Logic}}, {\sl \bibinfo{series}{Electronic
  Proceedings in Theoretical Computer Science}} \bibinfo{volume}{266}, pp.
  \bibinfo{pages}{315--328}, \doi{10.4204/EPTCS.266.20}.

\bibitemdeclare{article}{Hardy2011a}
\bibitem{Hardy2011a}
\bibinfo{author}{L.~\surnamestart Hardy\surnameend} (\bibinfo{year}{2011}):
  \emph{\bibinfo{title}{{Reformulating and Reconstructing Quantum Theory}}}.
\newblock \bibinfo{note}{\bibxiv{arXiv:1104.2066}}.

\bibitemdeclare{article}{Hardy2012Holism}
\bibitem{Hardy2012Holism}
\bibinfo{author}{L.~\surnamestart Hardy\surnameend} \&
  \bibinfo{author}{W.~\surnamestart Wootters\surnameend}
  (\bibinfo{year}{2012}): \emph{\bibinfo{title}{{Limited Holism and
  Real-Vector-Space Quantum Theory}}}.
\newblock {\sl \bibinfo{journal}{Foundations of Physics}}
  \bibinfo{volume}{42}(\bibinfo{number}{3}), pp. \bibinfo{pages}{454--473},
  \doi{10.1007/s10701-011-9616-6}.

\bibitemdeclare{inproceedings}{CatsOfRelations}
\bibitem{CatsOfRelations}
\bibinfo{author}{C.~\surnamestart Heunen\surnameend} \&
  \bibinfo{author}{S.~\surnamestart Tull\surnameend} (\bibinfo{year}{2015}):
  \emph{\bibinfo{title}{Categories of relations as models of quantum theory}}.
\newblock In: {\sl \bibinfo{booktitle}{{ Proceedings of the 12th International
  Workshop on} Quantum Physics and Logic}}, {\sl \bibinfo{series}{Electronic
  Proceedings in Theoretical Computer Science}} \bibinfo{volume}{195}, pp.
  \bibinfo{pages}{247--261}, \doi{10.4204/EPTCS.195.18}.

\bibitemdeclare{article}{selby2018reconstructing}
\bibitem{selby2018reconstructing}
\bibinfo{author}{J.~H. \surnamestart Selby\surnameend}, \bibinfo{author}{C.~M.
  \surnamestart Scandolo\surnameend} \& \bibinfo{author}{B.~\surnamestart
  Coecke\surnameend} (\bibinfo{year}{2018}):
  \emph{\bibinfo{title}{Reconstructing quantum theory from diagrammatic
  postulates}}.
\newblock \bibinfo{note}{\bibxiv{arXiv:1802.00367}}.

\bibitemdeclare{article}{selby2016process}
\bibitem{selby2016process}
\bibinfo{author}{John \surnamestart Selby\surnameend} \& \bibinfo{author}{Bob
  \surnamestart Coecke\surnameend} (\bibinfo{year}{2016}):
  \emph{\bibinfo{title}{Process-theoretic characterisation of the Hermitian
  adjoint}}.
\newblock {\sl \bibinfo{journal}{arXiv preprint arXiv:1606.05086}}.

\bibitemdeclare{article}{Selinger2007139}
\bibitem{Selinger2007139}
\bibinfo{author}{P.~\surnamestart Selinger\surnameend} (\bibinfo{year}{2007}):
  \emph{\bibinfo{title}{{Dagger Compact Closed Categories and Completely
  Positive Maps: (Extended Abstract)}}}.
\newblock {\sl \bibinfo{journal}{Electronic Notes in Theoretical Computer
  Science}} \bibinfo{volume}{170}, pp. \bibinfo{pages}{139--163},
  \doi{10.1016/j.entcs.2006.12.018}.

\bibitemdeclare{incollection}{selinger2011survey}
\bibitem{selinger2011survey}
\bibinfo{author}{P.~\surnamestart Selinger\surnameend} (\bibinfo{year}{2011}):
  \emph{\bibinfo{title}{{A survey of graphical languages for monoidal
  categories}}}.
\newblock In: {\sl \bibinfo{booktitle}{New structures for physics}},
  \bibinfo{publisher}{Springer}, pp. \bibinfo{pages}{289--355}.

\bibitemdeclare{article}{selinger2012finite}
\bibitem{selinger2012finite}
\bibinfo{author}{Peter \surnamestart Selinger\surnameend}
  (\bibinfo{year}{2012}): \emph{\bibinfo{title}{Finite dimensional Hilbert
  spaces are complete for dagger compact closed categories}}.
\newblock {\sl \bibinfo{journal}{arXiv preprint arXiv:1207.6972}}.

\bibitemdeclare{article}{thesis}
\bibitem{thesis}
\bibinfo{author}{S.~\surnamestart Tull\surnameend} (\bibinfo{year}{2018}):
  \emph{\bibinfo{title}{Categorical {O}perational {P}hysics. {DP}hil Thesis.}}
\newblock \bibinfo{note}{\bibxiv{arXiv:1902.00343}}.

\bibitemdeclare{article}{catreconstruction}
\bibitem{catreconstruction}
\bibinfo{author}{S.~\surnamestart Tull\surnameend} (\bibinfo{year}{2019}):
  \emph{\bibinfo{title}{A Categorical Reconstruction of Quantum Theory}}.
\newblock {\sl \bibinfo{journal}{Logical Methods in Computer Science}}.

\bibitemdeclare{article}{westerbaan2018dagger}
\bibitem{westerbaan2018dagger}
\bibinfo{author}{B.~\surnamestart Westerbaan\surnameend}
  (\bibinfo{year}{2018}): \emph{\bibinfo{title}{Dagger and dilations in the
  category of von {N}eumann algebras}}.
\newblock \bibinfo{note}{PhD Thesis. \bibxiv{arXiv:1803.01911}}.

\end{thebibliography}
\appendix

\section{Deriving the axioms} \label{sec:deriving-axioms}

Here we show how several of our axioms on $(\catC, \discard{}, \discardflip{})$ from section~\ref{sec:daggers-from-purification} follow from assumptions similar to the reconstruction~\cite{catreconstruction}. We begin with the pre-duals axiom. 
\begin{lemma}
Suppose that the sharpness and pure composition axioms hold, and $\discard{}$ is the unique effect sending every causal pure state to $\id{I}$. Then the pre-duals axiom holds iff each state $\discardflip{A}$ has a purification $\omega$ via some object $B$ which also purifies $\discardflip{B}$.
\end{lemma}
\begin{proof}

Let $\omega$ be any purification of $\discardflip{A}$. Then for all causal pure states $\psi$ of $A$ we have that 
\[
\peassign{
\left(
\input{app-proof-1.tikz}
\right)
}
 \ \ \ =  \ \ \  
 \input{app-proof-2.tikz}
\qquad
\text{ and so }
\qquad
\input{app-proof-3.tikz}
 \ \  \ = 
 \emptydiag 
\]
since the left-hand state above is pure and causal as $\peassign{\psi}$ is co-causal. Since $\psi$ was arbitrary, by assumption this establishes the third equation in~\eqref{eq:pre-dual}, and similarly the second holds when $\omega$ purifies $\discardflip{B}$.
\end{proof}

Now, we will call an object \emph{trivial} when its identity morphism satisfies
\[
\input{triv-ob.tikz}
\]
Following~\cite{catreconstruction} we say that $\catC$ satisfies \emph{pure exclusion} when for any pure state of a non-trivial object there is some non-zero effect $e$ with 
\[
\begin{tikzpicture}
	\begin{pgfonlayer}{nodelayer}
		\node [style=copoint] (0) at (-1.75, 0.75) {$e$};
		\node [style=point] (1) at (-1.75, -1) {$\psi$};
		\node [style=none] (2) at (0, 0) {$=$};
		\node [style=none] (3) at (1.5, 0) {$0$};
	\end{pgfonlayer}
	\begin{pgfonlayer}{edgelayer}
		\draw (1) to (0);
	\end{pgfonlayer}
\end{tikzpicture}

\]
Let us say that $\catC$ satisfies \emph{strong pure exclusion} when additionally the dual statement holds; every pure effect $\phi$ of a non-trivial object has $\phi \circ \rho = 0$ for some non-zero state $\rho$.

Next, recall that a \emph{kernel} of a morphism $f \colon A \to B$ is a morphism $\ker(f)$ with
\[
f \circ g = 0 \iff g = \ker(f) \circ h
\]
for a unique morphism $h$. Dually a \emph{cokernel} $\coker(f)$ satisfies $g \circ f = 0 \iff g = h \circ \coker(f)$ for some unique $h$. We will always take our (co)kernels to be (co)causal, and then they are unique up to (co)causal isomorphism whenever they exist. We will say that $\catC$ has \emph{split (co)kernels} when 
\begin{itemize}
\item Every morphism $f$ has both a causal kernel and co-causal cokernel;
\item For every causal kernel $k$ there is a unique co-causal cokernel $\peassign{k}$ with $\peassign{k} \circ k = \id{}$. Conversely, $k$ is the unique causal kernel for which this equation holds.
\end{itemize}

\begin{lemma} \label{lem:one}
Suppose that $\catC$ has split (co)kernels, satisfies strong pure exclusion, and has that $\discard{A} = 0 \implies \discardflip{A} = 0 \implies \id{A} = 0$ for all objects $A$. Then every causal pure state is a kernel and the sharpness axiom holds.
\end{lemma}
\begin{proof}
We first show that $\discard{} \circ f= 0 \implies f = 0$ for all morphisms $f$. Define $\img(f) := \ker(\coker(f))$. Then if $0 = \discard{} \circ f$ it follows that $0 = \discard{} \circ \img(f) = \discard{}$ and so $\img(f) = 0$. Since it factors over $\img(f)$, $f = 0$. 

We now show every causal pure state $\psi$ is a kernel. Let $i = \img(\psi) := \ker(\coker(\psi))$ with $i$ causal. Then $\psi = i \circ \phi$ for some $\phi$. Since $\peassign{i} \circ i = \id{}$ it is straightforward to check that $\phi$ is again pure and causal, and by the definition of $i$ that $f \circ \phi = 0 \implies f = 0$. Hence by pure exclusion the domain of $i$ is trivial, so that $\phi$ is a causal isomorphism, and $\psi = \img(\psi)$ as required.

Since our assumptions are self-dual, we dually have that every co-causal pure effect is a cokernel. Now for the sharpness axiom, by the split (co)kernels assumption, it suffices to check for each causal pure state $\psi$ that the co-causal cokernel effect $c := \peassign{\psi}$ is pure. Suppose that $\discard{} \circ f = c$. Then $\discard{} \circ f \circ \ker(c) = 0$ so $f \circ \ker(c) = 0$, and hence $f$ factors over $\coker(\ker(c)) = c$ since $c$ is a cokernel. Hence $f = \rho \circ c$ for some state $\rho$, which after applying $\discardflip{}$ we see is causal. Hence $c$ is indeed pure.
\end{proof}

As a result we may define an effect $\peassign{\psi}$ for each pure state $\psi$ as in~\eqref{eq:extended-state-dagger}. In this setting we can also capture our compositionality condition on pure states as follows. Say that $\catC$ has \emph{kernel composition} when (co)kernels are closed under $\otimes$, cokernels send pure states to pure states, and for all causal kernels $k$ and pure states $\psi$ we have 
\begin{equation} \label{eq:kernel-compact}
\peassign{ 
\left(
\begin{tikzpicture}
	\begin{pgfonlayer}{nodelayer}
		\node [style=map] (0) at (0, 0.5) {$\peassign{k}$};
		\node [style=point] (1) at (0, -1) {$\psi$};
		\node [style=none] (9) at (0, 1.75) {};
	\end{pgfonlayer}
	\begin{pgfonlayer}{edgelayer}
		\draw (1) to (9.center);
	\end{pgfonlayer}
\end{tikzpicture}

\right)
}
\ \ = \ \ 
\begin{tikzpicture}
	\begin{pgfonlayer}{nodelayer}
		\node [style=map] (0) at (0, -0.25) {$k$};
		\node [style=copoint] (1) at (0, 1.25) {$\peassign{\psi}$};
		\node [style=none] (9) at (0, -1.5) {};
	\end{pgfonlayer}
	\begin{pgfonlayer}{edgelayer}
		\draw (1) to (9.center);
	\end{pgfonlayer}
\end{tikzpicture}

\end{equation}

(Co)kernels are always closed under $\otimes$ in any compact category~\cite[Chap.~4]{thesis}, to which we refer along with~\cite{catreconstruction} for proof that all of these axioms hold in each theory $\Quant{S}$ for a phased field $S$.

\begin{lemma}
Suppose that $\catC$ has the properties of Lemma~\ref{lem:one} and kernel composition. Then $\catC$ satisfies the pure composition axiom.
\end{lemma}
\begin{proof}
By assumption and Lemma~\ref{lem:one} for any causal pure state $\phi$ and extra object the morphisms
\[
\begin{tikzpicture}
	\begin{pgfonlayer}{nodelayer}
		\node [style=none] (0) at (-2.25, 1) {};
		\node [style=none] (1) at (-2.25, -0.75) {};
		\node [style=point] (2) at (-0.75, -0.25) {$\phi$};
		\node [style=none] (3) at (-0.75, 1) {};
	\end{pgfonlayer}
	\begin{pgfonlayer}{edgelayer}
		\draw (0.center) to (1.center);
		\draw (3.center) to (2);
	\end{pgfonlayer}
\end{tikzpicture}

\qquad \qquad
\begin{tikzpicture}
	\begin{pgfonlayer}{nodelayer}
		\node [style=none] (0) at (-2.5, -0.75) {};
		\node [style=none] (1) at (-2.5, 1) {};
		\node [style=copoint] (2) at (-1, 0.5) {$\peassign{\phi}$};
		\node [style=none] (3) at (-1, -0.75) {};
	\end{pgfonlayer}
	\begin{pgfonlayer}{edgelayer}
		\draw (0.center) to (1.center);
		\draw (3.center) to (2);
	\end{pgfonlayer}
\end{tikzpicture}

\]
are a kernel and cokernel respectively. Hence the state in~\eqref{eq:pure-composition} is an application of a cokernel to a kernel and so is pure. Moreover, the above morphisms are (co)causal respectively and are inverse to eachother, so that for any pure bipartite state $\psi$ we obtain~\eqref{eq:pure-composition}. Finally, we need that effect in~\eqref{eq:pure-composition} is pure. By assumption it suffices to show that $c = \peassign{\psi}$ is pure for any pure state $\psi$, or in other words that for any causal such state $\psi$ and scalar $r$ that $r \circ c$ is pure. The proof is almost identical to the proof that $c$ is pure in the previous result, noting in the last step we have $f= \rho \circ c$ and then applying $\discardflip{}$ and normalisation we get $\rho = \sigma \circ r$ for some causal state $\sigma$, so that $f= \sigma \circ r \circ c$, as required.
\end{proof}

\begin{proof}[Proof of Example~\ref{ex:pure-properties-example}]
Let $\catC$ satisfy the operational principles of~\cite{catreconstruction}, to which we refer for the following facts. Essentially unique purification holds by assumption, noting that every causal isomorphism is pure and hence co-causal by the CP axiom, and normalisation since all scalars are pure and scalar multiplication is cancellative. 

Next, the sharpness axiom holds with $\peassign{\psi}$ being given by $\psi^\dagger$. Indeed, in~\cite{catreconstruction} it is shown that each such $\psi$ is a dagger kernel, giving~\eqref{eq:state-dagger-as-property}. Moreover if $\phi \circ \psi = \id{I}$ for some co-causal pure effect $\phi$ then since $\catC$ has addition with $\phi + \discard{} \circ \ker(\phi) = \discard{}$ for some effect $e$, and addition is cancellative, we obtain $\discard{} \circ \ker(\phi) \circ \psi = 0$ so that $\ker(\phi) \circ \psi = 0$. Hence since $\phi$ is a cokernel we get that $\psi$ factors over $\phi^\dagger$ and so both states are equal by assumption. Dually we see that $\psi$ is unique with $\psi^\dagger \circ \psi = \id{I}$ also. The pure composition and pre-dual axioms then hold by dagger compactness of $\catC$ and $\catC_\pure$.

It remains to verify identity tomography. Let $W$ be an endomorphism of some object $A$ which preserves all pure states $\psi$, with a purification $U$. Then for any such state since $\psi$ is pure we obtain that $U \circ \psi = \psi \otimes u_\psi$ for some causal pure state $u_\psi$, which is straightforward to check must be pure also. Now in~\cite{catreconstruction} it is shown that we may view $\catC_\pure$ as the quotient category of a dagger compact category $\catA (:= \plusI{\catB})$ given by identifying a group of unitary scalars. So let $V \in \catA$ with $U=[V]$ in $\catC_\pure$. Then in $\catA$ we again have $V \circ \psi = \psi \otimes v_\psi$ for some state $v_\psi$, for all states $\psi$. Causality of $[v_\psi]$ in $\catC_\pure$ gives that $v_\psi$ is an isometry. We now show that $v_\psi$ does not depend on $\psi$. 

Since $\catA$ has an addition operation on morphisms we then have for all states $\psi, \phi$ that
\[
\psi \otimes v_\psi + \phi \otimes v_\phi
=
V \circ \psi + V \circ \phi 
=
V \circ (\psi + \phi)
=
(\psi + \phi) \otimes v_{\psi + \phi}
=
\psi \otimes v_\psi + \phi \otimes v_\phi
\]
Whenever $\psi$ and $\phi$ are orthonormal, so that $\psi^\dagger \circ \phi = 0$ and $\psi^\dagger \circ \psi = \id{I} = \phi^\dagger \circ \phi$, by composing with $\psi^\dagger$ and $\phi^\dagger$ it then follows that $v_\psi = v_{\psi + \phi} = v_\phi$. But more generally for we may write $\psi = \phi \circ z + \mu$ for a state $\mu$ which is orthogonal to $\phi$, where $z = \psi^\dagger \circ \phi$. It then follows that $v_\phi = v_\mu = v_\psi$. Since every state in $\catA$ is a multiple of an isometric one, the same holds for arbitrary states $\psi, \phi$.

Hence we have $V \circ \psi = \psi \otimes v$ for some fixed isometric state $v$. But since $\catA$ is well-pointed, this gives $V = \id{} \otimes v$. Then $[U] = \id{} \otimes [v]$ also, so that $W=\id{}$ as required.
\end{proof}

\end{document}